\documentclass[12pt]{article}
\usepackage{graphicx}
\usepackage{url} 
\usepackage{graphicx} 
\usepackage{setspace, amsmath, mathtools, enumitem, amsfonts, amsthm}
\usepackage{natbib}
\usepackage{bbm}
\usepackage{tikz}
\newtheorem{theorem}{Theorem}

\usepackage{subcaption}
\usepackage{enumitem}
\usepackage{hyperref}
\usepackage[nameinlink]{cleveref}
\usepackage{hyperref}
\hypersetup{
    colorlinks=true,
    linkcolor=red,
    filecolor=magenta,  
    citecolor=blue,
    urlcolor=red,
    pdftitle={Overleaf Example},
    pdfpagemode=FullScreen,
}

\newlist{assumptionnum}{enumerate}{1}
\setlist[assumptionnum]{label=(\alph*),ref=\thetheorem(\alph*)}

\newcommand{\Cwsd}{Counterbalanced within-subjects design}
\newcommand{\cwsd}{counterbalanced within-subjects design}
\newcommand{\Wsd}{Within-subjects design}
\newcommand{\bsd}{between-subjects design}

\newtheorem*{remark}{Remark}
\newtheorem*{retheorem}{Theorem}
\newtheorem{prop}{Proposition}
\newtheorem*{reprop}{Proposition}
\newtheorem{cor}{Corollary}
\newtheorem*{recor}{Corollary}
\newtheorem{definition}{Definition}

\newtheorem{assumption}{Assumption}

\newcommand{\indep}{\hspace{-.2em}\perp\!\!\!\!\perp\hspace{-.2em}} 

\newcommand{\blind}{0}

\begin{document}

\def\spacingset#1{\renewcommand{\baselinestretch}%
{#1}\small\normalsize} \spacingset{1}


\if0\blind
{
  \title{\bf Causal Inference in Counterbalanced Within-Subjects Designs}
  \author{Justin Ho
  \thanks{Primary and Corresponding Author (jho@g.harvard.edu)}\hspace{.2cm}\\
    Harvard University\\
    and \\
    Jonathan Min\\
    University of California, Berkeley}
  \maketitle
} \fi

\if1\blind
{
  \bigskip
  \bigskip
  \bigskip
  \begin{center}
    {\LARGE\bf Causal Inference in Counterbalanced Within-Subjects Designs}
\end{center}
  \medskip
} \fi

\bigskip

\begin{abstract}
    Experimental designs are fundamental for estimating causal effects. In some fields, within-subjects designs, which expose participants to both control and treatment at different time periods, are used to address practical and logistical concerns. Counterbalancing, a common technique in within-subjects designs, aims to remove carryover effects by randomizing treatment sequences. Despite its appeal, counterbalancing relies on the assumption that carryover effects are symmetric and cancel out, which is often unverifiable \emph{a priori}. In this paper, we formalize the challenges of counterbalanced within-subjects designs using the potential outcomes framework. We introduce \emph{sequential exchangeability} as an additional identification assumption necessary for valid causal inference in these designs. To address identification concerns, we propose diagnostic checks, the use of washout periods, and covariate adjustments, and alternative experimental designs to \cwsd{}. Our findings demonstrate the limitations of counterbalancing and provide guidance on when and how within-subjects designs can be appropriately used for causal inference.
\end{abstract}

\noindent%

\vfill

\newpage
\spacingset{1.45} 

\section{Introduction}

\begin{quote}
\textit{"Failure to conclude that a model is false must be a failure of our imagination, not a success of the model."}
\hfill --- \cite{mcelreath2018statistical}
\end{quote}

Experimental designs are increasingly used in the social and behavioral sciences in estimating the causal effects of various interventions. The randomized controlled trial (RCT) is widely regarded to be the `gold standard' when it comes to estimating the causal effects of treatment as opposed to estimates from an observational study. Typically, participants are divided into distinct treatment and control groups, allowing for straightforward comparisons of intervention outcomes in RCTs.\footnote{This design is called \bsd{} in some fields since interventions are randomly assigned between subjects \citep{02_maxwell2017designing}.} However, this conventional approach is not the only experimental design used.

\Wsd{}, also known as repeated-measures design, offers an alternative that exposes all participants to every condition in a study. This design is sometimes used in fields such as experimental psychology \citep{keren2014between}, management sciences \citep{Lane2024}, political science \citep{09_Clifford}, and even medical trials \citep{Sarkies2019}, which enables researchers to compare responses across treatments for the same individuals. Experimenters often employ these methods when they face logistical challenges in recruitment, implementation (particularly regarding the costs associated with specific trials), or when the population concerned is relatively small, and the experimenter wishes to ensure there is \textit{common support} in terms of the exposure of the treatment with respect to the covariates (\citealp{02_maxwell2017designing}; \citealp{ErlebacherAlbert1977Daao}). However, this experimental design introduces challenges in causal inference. 

For example, imagine a study testing the effect of background music on concentration. Participants complete two problem-solving tasks: one in silence and one with music. If every participant always does the silent condition first and the music condition second, any observed difference in performance might not be due to the music itself but rather to other factors. For instance, they might perform better in the second task simply because they have become more familiar with the problem format (a practice effect) or worse due to mental fatigue. Since these order-related influences are confounded with the treatment, the researcher cannot tell whether any observed changes are truly caused by the music or just by the passage of time and repeated testing. The chief culprit confounding the estimate of the treatment effect in within-subjects designs are \emph{carryover effects}, which are the residual effects from the first period that influence responses in the second period \citep{02_maxwell2017designing}.

To address these concerns, counterbalancing is often applied to within-subjects designs. In clinical trials, this design is also known as crossover trials or AB/BA designs \citep{ZHANG2022_crossovertrials, sibbald1998understanding, matthew_multiperiod_crossover}. This approach randomizes the order of treatment and control conditions for participants, assuming that this would cancel out any carryover effects, resulting in an unbiased estimate. Figure \ref{fig:CWSD_Figure} shows how a \cwsd{} is implemented in the case where there are two experimental conditions. However, this assumption is often overly optimistic. Counterbalancing assumes that carryover effects are symmetric and cancel out, yet this symmetry is unverifiable \emph{a priori}. Differential carryover effects, where one condition's impact persists or interacts with subsequent conditions in an asymmetric manner, can bias estimates of the average treatment effect even when participants are properly counterbalanced.

There are some advantages in using within-subjects design in some empirical settings, as it aligns more naturally with economic and psychological theories that model how individuals respond to changing conditions. For example, in empirical studies testing utility theories and preference, participants might exposed to multiple conditions of the study, allowing the experimenter to observe how an individual makes such a tradeoff \citep{09_Charness}. 

Despite their inferential limitations, counterbalanced within-subjects designs offer several practical advantages. They allow researchers to collect multiple observations per participant, improving statistical efficiency—especially useful when working with small populations or when recruitment is costly or difficult, such as in studies involving venture capitalists \citep{Lane2024}. In such cases, a between-subjects design may suffer from issues like limited overlap or lack of common support across covariates, making within-subjects designs an attractive alternative despite their potential identification problems.

In this paper, we analyze the implications of using counterbalanced within-subjects designs for causal inference, analyzing the problem with the potential outcomes framework. Specifically, we first explore the assumptions required in \cwsd{} to produce treatment effect estimates \emph{consistent} with those obtained from \bsd{} in the case of two discrete time periods, with the goal of estimating the Average Treatment Effect (ATE) as defined in the potential outcomes framework. We then introduce \textit{sequential exchangeability} as a generalized assumption required, alongside the standard assumptions of the Rubin Causal Model, to identify the causal effects of treatment in other sequential designs similar to a counterbalanced within-subjects design. We introduce sequential randomization and selective sequential randomization in Section \ref{sec: seq_exchangeability_alt_design} as alternative sequential designs where the assumptions of \textit{sequential exchangeability} is more credible than the commonly used \cwsd{}.

Sequential designs are often used to improve the efficiency and reliability of experiments. For example, \cite{04_whitehead2020estimation} use a design that allows dropping underperforming treatments mid-trial, \cite{08_zhou2018sequential} propose a rerandomization method to improve covariate balance over time, and \cite{03_tamura2011estimation} apply a sequential parallel design to reduce placebo effects. While these approaches differ in motivation, they all rely on using data collected across stages to better estimate treatment effects. We focus on how similar ideas apply to counterbalanced within-subjects designs and relate \textit{sequential exchangeability} as a condition under which such designs can identify the average treatment effect (ATE).

We want to emphasize that such assumptions are typically difficult to verify \emph{a priori} and hence, we do not advocate the use of such experimental designs to estimate ATE, even in the case of sequential randomization where the assumptions are more credible. However, in situations where such a design is required due to the aforementioned logistical challenges, we provide a way to reason, and guidance to ensure that the treatment effects are properly estimated.

Our work builds on existing work of analyzing experimental designs within the potential‐outcomes framework such as in mediation analysis, factorial designs, and conjoint studies, by making explicit the precise causal assumptions each design invokes. In mediation analysis, for example, researchers decompose a treatment’s total effect into direct and indirect pathways under a sequential ignorability assumption that links potential mediators and outcomes \citep{imai2013experimental,11_imai2010general}. Conjoint analysis, meanwhile, treats each attribute as a component ``treatment,'' isolating marginal contributions and interactions via randomized component assignment and carefully articulated exclusion restrictions \citep{egami2019causal}. In the same spirit, our paper examines counterbalanced within-subjects design to ask: what is the set of causal assumptions required to point‐identify the ATE in counterbalanced within‐subject designs? By iterating these identifying assumptions in full generality, we extend the potential‐outcomes framework to reason for arbitrary, higher‐order carryover and learning effects, rather than relying on the narrow parametric or symmetry conditions typical of two‐period crossover estimators \citep{23_AnalysisOfCrossoverTrial}.

The paper is organized as follows. Section \ref{sec: identification_assumption} outlines the identification assumptions required in a counterbalanced within-subjects design and provides illustrative examples of when these assumptions might be violated. Section \ref{sec: seq_exchangeability_alt_design} generalizes the assumptions and introduce the concept of sequential exchangeability, as well as alternative sequential designs where the assumptions are more credible. Section \ref{sec: what_can_we_do} discusses heuristic checks to verify whether violations of the assumptions occur and explores ways to recover an unbiased causal estimate using simulations under certain additional assumptions. Section \ref{sec: conclusion} concludes the paper.

\section{Identification in Counterbalanced Within-Subjects Design}
\label{sec: identification_assumption}

In this section, we formalize the identification challenges that arises from \cwsd{} by comparing it with an experiment using a \bsd{}, with the goal of recovering the Average Treatment Effect (ATE). We develop the additional assumptions required so that the estimates in a \cwsd{} will yield the same results as in a \bsd{} in this section.

\subsection{Notation}

 Assume that we have $n$ units of observations, where we have the set of associated observed outcome $\{Y_i\}_{i=0}^n$ and an treatment $\{Z_i\}_{i=0}^n$. In the potential outcomes framework, we imagine jointly observing both outcomes of treatment and control \citep{10_holland1986statistics}. We denote the potential outcomes as $(Y_i(1), Y_i(0))$, where $1$ and $0$ indicate that the unit received treatment or control, respectively. From these definitions, we can define the key identity 
\begin{equation}
    Y_i = Y_i(1)Z_i + Y_i(0)(1 - Z_i)
\end{equation}
which captures the relationship between potential outcomes and observed outcomes. Note that by this identity, $Y_i$ denotes the observed outcome of unit $i$. We also observe the pre-treatment covariates, denoted as $\{X_i\}_{i=0}^n$.

Per the standard definition of the Average Treatment Effect (ATE), it is defined as the difference in the potential outcomes of treatment and control\[\tau = \mathbb{E}[Y_i(1) - Y_i(0)]\] 
Under the standard assumptions of the Rubin Causal Model, the Average Treatment Effect is identified if \citep{10_holland1986statistics}
\begin{enumerate}
    \item $(Y_i(0), Y_i(1)) \indep Z_i \mid X_i$ \textit{(Ignorability)}
    \item $0 < \text{Pr}\{Z_i = 1 \mid X_i = x\} < 1 \quad \forall x \in X$ \textit{(Overlap)}
    \item Stable Unit Treatment Value Assumption (SUTVA) - treatment received by a unit does not affect the potential outcomes of other units and the treatment is consistent (there is no variations in the application of the treatment)
\end{enumerate}

These assumptions together imply that the unobserved counterfactuals can be recovered, in expectation, from observed outcomes conditional on covariates. This leads to the identifiability of the Average Treatment Effect in a between-subjects design where the assumptions are plausibly satisfied.

\begin{theorem}
\label{theorem: identify_bsd}
    Under the standard assumptions of the Rubin Causal Model, ATE is identifiable in a between-subjects design.
\end{theorem}

One of the biggest appeals of an experimental design such as a \bsd{} is that participants are randomly assigned to receive either treatment or control but not both. This removes selection bias if done properly, allowing us to reasonably assume that ignorability holds, and therefore, a simple difference in means estimator without adjustment would lead to an unbiased estimate of ATE.

\begin{remark}
    We will refer to $\tau = \mathbb{E}[Y_i(1) - Y_i(0)]$ as ATE, or Average Treatment Effect in the rest of the paper. For `time-specific' Average Treatment Effect, we will explicitly mention that we are referring to a different parameter of concern since they entail quite a different interpretation.
\end{remark}

\subsection{Structure of a Counterbalanced Within-Subjects Design}
\label{sec: assumption_cwsd}

We begin the analysis of a \cwsd{} in the case where there are two experimental conditions, treatment and control, and two discrete time periods, $t_1$ and $t_2$. In a \cwsd{}, each unit of observation is assigned to both treatment and control, but with each unit being randomly assigned to the sequence of treatment. Let $S_i$ be the sequence a unit $i$ is assigned to, where $\{S_i\}^n_{i=1}$ for when $S_i = 0$, unit $i$ receives control at $t_1$ and then receives treatment at $t_2$ and when $S_i=1$ unit $i$ receives treatment at $t_1$ and then receives control at $t_2$.

\begin{figure}
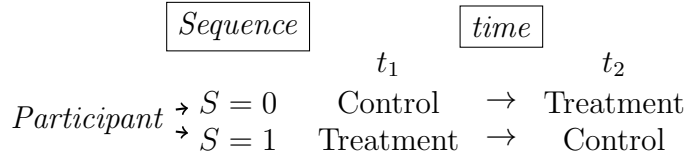

    \centering
    \tikz{
    \node (2a) at (2,0) {$t_1$};
    \node (3a) at (3,0) {};
    \node (4a) at (5,0) {$t_2$};

    \node (1b) at (0,-.5) {$S=0$};
    \node (2b) at (2,-.5) {Control};
    \node (3b) at (3.5,-.5) {$\rightarrow$};
    \node (4b) at (5,-.5) {Treatment};

    \node (1c) at (0,-1) {$S=1$};
    \node (2c) at (2,-1) {Treatment};
    \node (3c) at (3.5,-1) {$\rightarrow$};
    \node (4c) at (5,-1) {Control};

    \node [rectangle,draw] (S) at (0,0.5) {\textit{Sequence}};
    \node  [rectangle,draw] (T) at (3.5,0.5) {\textit{time}};

    \node (P) at (-2, -0.75) {\textit{Participant}};

    \path[->, thick] (P) edge (1b);
    \path[->, thick] (P) edge (1c);
    }
    \caption{An illustration of a counterbalanced within-subjects experimental design. Participants are randomized into one of the sequences, and then exposed to treatment and control at different times. }
    \label{fig:CWSD_Figure}
\end{figure}

$Z_{i,t_1} \in \{0,1\}$ indicates that unit $i$ receives control or treatment at $t_1$, where $Z_{i,t_1} = 1$ when unit $i$ receives treatment and 0 otherwise. In a counterbalanced within-subjects design, $Z_{i,t_2} = 1 - Z_{i,t_1}$ and $S_i = Z_{i,t_1}$. We can define the observed outcomes $Y_{i, t_1}^{\text{obs}}$ and $Y_{i, t_2}^{\text{obs}}$ in terms of the potential outcomes as
\begin{equation}
    Y_{i,t_1}^{\text{obs}} = Y_{i,t_1}(1)Z_{i,t_1} + Y_{i,t_1}(0)(1-Z_{i,t_1})
\end{equation}

\begin{equation}
\label{equation: t2_identity_cwsd} 
    Y_{i,t_2}^{\text{obs}} = Y_{i,t_2}(0,1)(1-Z_{i,t_1}) + Y_{i,t_2}(1,0)Z_{i,t_1}
\end{equation}

At $t_1$, the design is equivalent to a between-subjects design, with the potential outcomes at $t_1$, denoted by $Y_{i,t_1}(1)$ for treatment and $Y_{i,t_1}(0)$ for control. The average treatment effect (ATE) at $t_1$ is defined as
\[
\tau_{t_1} = \mathbb{E}\big[Y_{i,t_1}(1) - Y_{i,t_1}(0)\big].
\]
Since participants in a \cwsd{} are randomized into the sequence in which they receive the treatment, $\tau_{t_1} = \tau$ trivially as a result.

At $t_2$, \cwsd{} deviates from the standard framework. As a result of the sequential treatment design, participants switch conditions so that the observed potential outcomes are:
\begin{itemize}
    \item $Y_{i,t_2}(0,1)$ for those who received control at $t_1$ and treatment at $t_2$, and 
    \item $Y_{i,t_2}(1,0)$ for those who received treatment at $t_1$ and control at $t_2$.
\end{itemize}
Importantly, the potential outcomes $Y_{i,t_2}(0,0)$ and $Y_{i,t_2}(1,1)$ are never observed for all $i$ in a \cwsd{}. Thus, the ``average treatment effect'' at $t_2$, denoted as $\tau_{t_2}$, in a \cwsd{} is defined as
\[
\tau_{t_2} = \mathbb{E}\big[ Y_{i,t_2}(0,1) - Y_{i,t_2}(1,0) \big].
\] 
Note that this is not the typical parameter that people usually think of when they say ATE, since this is more accurate to describe it as the average treatment effect at $t_2$ given the treated was untreated at time $t_1$ and the untreated was treated at time $t_1$.

Unlike $\tau_{t_1}$, the equality $\tau_{t_2} = \tau$ does not hold automatically. In the following section, we develop the assumptions required such that 
\[
\tau_{t_2} = \mathbb{E}\big[ Y_{i,t_2}(0,1) - Y_{i,t_2}(1,0) \big] = \mathbb{E}\big[ Y_{i}(1) - Y_{i}(0) \big] = \tau
\]
since this is often the purpose of using \cwsd{} in causal inference. In the following sections, we will primarily focus on the analysis of $\tau_{t_2}$.

\begin{remark}
    Note that $\tau_{sequence} = \mathbb{E}\big[ Y_{i}(0,1) - Y_{i}(1,0) \big]$ is a valid parameter/estimand. For example, an experimenter might be interested in whether the sequence of control followed by treatment versus treatment followed by control leads to a difference in outcomes. For instance, an experimenter wishing to understand the effects of teaching New Math at an early age followed by a more standard math curriculum—compared to the reverse sequence—might investigate its effect on students' math ability might be interested in estimating $\tau_{sequence}$ without reducing it to the standard definition of ATE.
\end{remark}

\subsection{Estimating Average Treatment Effects in a Counterbalanced Within-Subjects Design}

To justify our focus on \(\tau_{t_2}\), consider a common estimation method used to estimate ATE in a counterbalanced within‐subjects design that employs the following regression specification with an OLS estimator:
\begin{equation}
\label{eq: OLS_estimate}
Y_{i,t} = \alpha_0 + \alpha_1 Z_{i,t} + \alpha_2 \mathbbm{1}\{t = t_1\} + \alpha_3 X_{i,t} + \varepsilon_{i,t}.
\end{equation}

A key motivation for using a \cwsd{} is to increase effective sample sizes by pooling observations across time periods. In this context, the pooled OLS estimate \(\hat{\alpha}_1\) is interpreted as the estimated average treatment effect in a \cwsd{}, that is, \(\hat{\tau}^{\text{CWSD}} = \hat{\alpha}_1\).

\begin{prop}
\label{prop: FWL}
The OLS estimate from Equation \eqref{eq: OLS_estimate}, \(\hat{\alpha}_1 = \hat{\tau}^{\text{CWSD}}\), can be expressed as a convex combination of the period-specific estimates:
\[
\hat{\tau}^{\text{CWSD}} = q\,\hat{\tau}_{t_1} + (1-q)\,\hat{\tau}_{t_2},
\]
with \(q \in [0,1]\).
\end{prop}

Proposition \ref{prop: FWL} shows that the pooled OLS estimate from a counterbalanced within-subjects design is a weighted average of the treatment effects from each period. This decomposition allows us to interpret the estimator as a blend of the two time-specific effects, with the weight \(q\) depending on the variance structure of the treatment assignment across periods.

\begin{cor}
If \(\mathbb{E}[\hat{\tau}_{t_2}] = \tau\), then \(\mathbb{E}[\hat{\tau}^{\text{CWSD}}] = \tau\).
\end{cor}

This comes immediately from the linearity of expectations, which justifies our subsequent focus on the analysis of $\tau_{t_2}$.

\subsection{Identification Assumptions in a Counterbalanced Within-Subjects Design}

\begin{figure}[h]
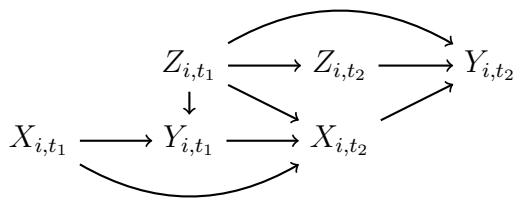

    \centering
    \tikz{
    \node (space) at (0,1) {};
    \node (space) at (0,-2) {};
    \node (Zt1) at (2,0) {$Z_{i,t_1}$};
    \node (Xt1) at (0,-1) {$X_{i,t_1}$};
    \node (Yt1) at (2,-1) {$Y_{i,t_1}$};
    
    \node (Zt2) at (4,0) {$Z_{i,t_2}$};
    \node (Xt2) at (4,-1) {$X_{i,t_2}$};
    \node (Yt2) at (6,0) {$Y_{i,t_2}$};

    \path[->, thick] (Xt1) edge (Yt1);
    \path[->, thick] (Zt1) edge (Yt1);
    \path[->, thick] (Zt1) edge (Zt2);
    \path[->, thick] (Yt1) edge (Xt2);
    \path[->, thick] (Zt2) edge (Yt2);
    \path[->, thick] (Xt2) edge (Yt2);
    \path[->, thick] (Zt1) edge[out=30, in=150] (Yt2);
    \path[->, thick] (Zt1) edge (Xt2);
    \path[->, thick] (Xt1) edge[out=-30, in=-150] (Xt2);
    }
    \caption{Potential Causal Pathways in a \Cwsd{}. Note that $X_{i,t_1}, X_{i,t_2}$ includes both \emph{observed and unobserved} covariates for concision.}
    \label{fig:CWSD_DAG}
\end{figure}

One of the major challenges in estimating the average treatment effect (ATE) in a \cwsd{} is the presence of \emph{carryover effects}. These effects refer to residual influences from the treatment (or control) administered in the first period, denoted by \(z_{t_1}\), on the outcome observed at the second period, beyond the effect of the treatment received at \(z_{t_2}\). Such residual effects can introduce unwanted variability or bias in the estimation of the treatment effect at time \(t_2\).

Figure \ref{fig:CWSD_DAG} illustrates the causal pathways in a \cwsd{}. The principal identification challenge arises at \(t_2\), where the experimental design allows treatment received at \(t_1\) to confound the effect of the treatment at \(t_2\).

The first potential source of bias is due to the \emph{direct carryover effects} of the treatment at \(t_1\) on the outcome at \(t_2\). This effect is captured by the causal pathway 
\[
Z_{i,t_1} \to Y_{i,t_2},
\]
indicating that any residual impact of the treatment (or control) at \(t_1\) can directly influence the observed outcome \(Y_{i,t_2}^{\text{obs}}\) at \(t_2\). Before we formalize the concept of direct carryover effects, we first have to define the potential outcome $Y_{i,t_2}(z_{i,t_2})$.

\begin{definition}[Potential Outcome Without Prior Exposure]
    $Y_{i,t_2}(z_{i,t_2})$ is the potential outcome of unit $i$ if unit $i$ was recruited at $t_2$ and exposed to treatment $z_{i,t_2}$ without any prior exposure at $t_1$.
\end{definition}

Note that in a \cwsd{}, this is \emph{never} observed. However, it serves as a useful counterfactual for reasoning about what the outcome would have been if unit $i$ appeared at $t_2$ and was only assigned to $z_{i,t_2}$ without any prior exposure. An implicit assumption in the way $Y_{i,t_2}(z_{i,t_2})$ is defined is that \[Y_{i,t_2}(z_{i,t_2}) \indep Z_{i,t_1} \mid Z_{i,t_2}, X_{i,t_2}\] since $Z_{i,t_1}$ `does not exist' in this counterfactual scenario. With this, we can define direct carryover effects as follows:

\begin{definition}[Direct Carryover Effects]
\label{defi:direct_carryover_effects}
For any \(i\) and for each combination of treatment assignments \(z_{t_1}, z_{t_2} \in \{0,1\}\), the \emph{direct carryover effect} of \(z_{t_1}\) on the potential outcome \(Y_{i,t_2}(z_{t_1}, z_{t_2})\) given \(z_{t_2}\) is defined as the difference between the potential outcome \(Y_{i,t_2}(z_{t_1}, z_{t_2})\) and the potential outcome without prior exposure \(Y_{i,t_2}(z_{t_2})\)
\[
C_{i, z_{t_1} \to z_{t_2}} \coloneqq Y_{i,t_2}(z_{t_1}, z_{t_2}) - Y_{i,t_2}\bigl(z_{t_2}\bigr) \mid Z_{i,t_1} = z_{t_1}, Z_{i,t_2} = z_{t_2}, X_{i,t_2}
\]
\end{definition}

In words, the direct carryover effects as stated in Definition \ref{defi:direct_carryover_effects} measures the difference between the potential outcome when an individual receives the sequence \((z_{t_1}, z_{t_2})\) and the potential outcome if only the treatment at \(t_2\) were administered. In the perspective of the potential outcomes framework, we can imagine that unit $i$ was recruited at time $t_2$ without participating at all at time $t_1$, and has the potential outcome $Y_{i,t_2}\bigl(z_{t_2}\bigr)$.

In a \cwsd{}, counterbalancing is typically implemented under the assumption that these direct carryover effects are \emph{simple}—that is, they are symmetric and cancel out in expectation \citep{02_maxwell2017designing}. Formally, we define simple direct carryover effects as follows.

\begin{assumption}[Simple Direct Carryover Effects]
\label{assumption:simple_carryover_effects}
Direct carryover effects are \emph{simple} if, for any distinct treatment assignments \(z_1, z_2 \in \{0,1\}\) where \(z_1 \neq z_2\), we have
\[
\mathbb{E}\Bigl[C_{i, z_1 \to z_2}\Bigr] = \mathbb{E}\Bigl[C_{i, z_2 \to z_1}\Bigr] = C
\]
for some constant \(C \in \mathbb{R}\).
\end{assumption}

In addition to direct effects, \emph{indirect carryover effects} may occur when the treatment at \(t_1\) influences covariates measured at \(t_2\), which in turn affect the outcome \(Y_{i,t_2}\). As shown in Figure \ref{fig:CWSD_DAG}, there are two potential indirect pathways:
\begin{enumerate}[label=(\arabic*), leftmargin=*]
    \item \(Z_{i,t_1} \to Y_{i,t_1} \to X_{i,t_2} \to Y_{i,t_2}\),
    \item \(Z_{i,t_1} \to X_{i,t_2} \to Y_{i,t_2}\).
\end{enumerate}
To ensure an unbiased estimate of the treatment effect at \(t_2\), these indirect pathways must be blocked, typically by controlling for \(X_{i,t_2}\) or with randomization at $t_2$. 

Finally, given that outcomes may change over time irrespective of treatment, it is necessary to account for time-specific shifts. In other words, even in the absence of treatment, the evolution of outcomes between \(t_1\) and \(t_2\) should be similar across treatment and control groups. This is analogous to the common parallel trends assumption in difference-in-differences analysis, and is stated formally as
\[
Y_{i,t_2}(z_{i,t_2}) = Y_{i,t_1}(z_{i,t_2}) + c, \quad \forall z_{i,t_2} \in \{0,1\}, \quad c \in \mathbb{R}.
\]
This assumption implies that any time-related change in the outcome is the same for both treatment and control conditions, differing only by a constant \(c\).

The direct and indirect carryover effects, along with the parallel trends assumption, comprise the additional assumptions required for the identification of the ATE in a \cwsd{}. We summarize these conditions in the following theorem.

\begin{theorem}
\label{theorem: identify_cwsd}
Under the standard Rubin Causal Model assumptions (SUTVA and overlap), the Average Treatment Effect (ATE) is identifiable in a \cwsd{} at $t_2$ if the following assumptions hold:
\begin{enumerate}
    \item \textbf{Simple Direct Carryover Effects:} As defined in Assumption \ref{assumption:simple_carryover_effects}.
    \item \textbf{Ignorability at Time \(t_2\):} 
    \[
    Y_{i,t_2}(z_{i,t_2}) \indep Z_{i,t_2} \mid X_{i,t_2}, \quad \forall z_{i,t_2} \in \{0,1\}.
    \]
    \item \textbf{Parallel Trends:} 
    \[
    Y_{i,t_2}(z_{i,t_2}) = Y_{i,t_1}(z_{i,t_2}) + c, \quad \forall z_{i,t_2} \in \{0,1\}, \quad c \in \mathbb{R}.
    \]
\end{enumerate}
\end{theorem}

\subsection{Violations of the Identification Assumptions} 

If we assume that direct carryover effects follow the ``simple" definition in Definition \ref{assumption:simple_carryover_effects}, then we have \(\mathbb{E}[C_{i,0 \to 1}] = \mathbb{E}[C_{i,1 \to 0}]\). This balance allows us to obtain an unbiased estimate of the treatment effect at time \( t_2 \), as shown in Theorem \ref{theorem: identify_cwsd}.  

However, when carryover effects are not simple — meaning \(\mathbb{E}[C_{i,0 \to 1}] \neq \mathbb{E}[C_{i,1 \to 0}]\) — counterbalancing alone cannot solve the issue. To illustrate why, consider the following (somewhat extreme) example.  

Perder Pharma, a pharmaceutical company, is developing a pain-relief drug called Moxycontin, which is opioid-based and highly effective at reducing pain (denoted as some \(\tau < 0\)). To save costs, they use a \cwsd{} to test its effect on back pain, randomizing participants into different sequences: some receive Moxycontin first, while others get it later.  

For those who receive a placebo first, the carryover effect \(C_{0 \to 1}\) is close to zero—these participants are fine until they later take Moxycontin, which genuinely helps with pain. However, the carryover effect in the other direction, \(C_{1 \to 0}\), is much larger than \(|\tau|\) and greater than zero. This is because Moxycontin is an opioid, and withdrawal causes severe pain, making things much worse for participants who are taken off the drug. The withdrawal is so extreme that the FDA ultimately concludes Moxycontin increases back pain and denies Perder's application. As a result, more than 400,000 lives are saved, and many others avoid the devastation of opioid addiction.  

Beyond direct carryover effects, \cwsd{} also lacks full randomization in the second time period, unlike \bsd{}. This is because the treatment condition in the first period, \(Z_{i,t_1}\), can indirectly influence covariates—whether observed or unobserved—at \(t_2\). If these covariates (\(X_{i,t_2}\)) are not accounted for, the assumption of ignorability at \(t_2\) may be violated. To illustrate how this assumption may fail to hold, we provide two examples (1) induced selection bias and (2) survivorship bias. 

In the case of an \emph{induced selection bias}, consider an experiment testing the effect of Ozempic on blood sugar levels. A key covariate influencing blood sugar is physical activity. If participants receiving Ozempic are motivated to increase their physical activity---perhaps due to perceived health improvements or weight loss---while the control group’s activity levels remain unchanged, this introduces imbalance. At time \(t_2\), the Ozempic group may, on average, have higher physical activity levels, which would confound the estimated treatment effect and violate the assumption of ignorability.

As for \emph{survivorship bias}, consider an experiment testing a treatment for cancer patients. At \(t_1\), participants in the control group who have worse health outcomes might die, effectively dropping out at \(t_2\), whereas some in the treatment group who would have died if assigned to control at \(t_1\) might survive and be observed at \(t_2\). This leads to participants in the treatment group being relatively healthier at \(t_2\) compared to their counterparts in the control group, assuming that the participants in the treatment group at $t_2$ does not suffer a decline of health due to not receiving the treatment at $t_1$.

Lastly, the parallel trends assumption may also be violated. Suppose a study involves participants solving complex problems in two sessions, using a counterbalanced design. If receiving an innovative problem-solving prompt in the first session not only improves performance in that session but also fundamentally changes how participants approach future problems, their learning trajectory shifts. As a result, the outcome evolution between sessions is nonparallel, violating the parallel trends assumption.  

\section{Generalization \& Sequential Exchangeability}
\label{sec: seq_exchangeability_alt_design}

In this section, we develop the concept of \emph{sequential exchangeability}, extending the original assumptions stated in Theorem \ref{theorem: identify_cwsd}. This extension is motivated by the observation that carryover effects are not unique to \cwsd{} but can also influence inference in other sequential experimental designs.

For instance, in \cite{09_Clifford}, the authors advocated for the adoption of pre-post designs, wherein participants are initially exposed to a control condition before transitioning to a between-subjects design. They posited that this approach enhances precision. However, in one of the experiments they replicated, carryover effects appeared to influence inference—not necessarily in the direction of the effect, but rather in its magnitude.

Experiment 2 in \cite{09_Clifford} replicated a landmark study on foreign aid by \citep{GilensMartin2001PIaC}, investigating the influence of factual information regarding federal budget allocation for foreign aid on public opinion. In that study, respondents were informed that spending on foreign aid constituted less than 1\% of the federal budget, after which they were asked whether foreign aid spending should be increased or decreased, with responses recorded on a five-point scale.

To compare different experimental designs, participants were randomly assigned to one of two conditions: the \textit{post-only} design or the \textit{pre-post} design. In the \textit{post-only} design, participants were randomized to receive the budgetary information, consistent with a between-subjects approach. Conversely, in the \textit{pre-post} design, participants first provided their opinion on foreign aid without the budgetary information, were subsequently exposed to the budgetary information, and finally responded to the same question again. The results of this replication indicate substantial differences in the estimated effects, as depicted in Figure \ref{fig:difference-in-estimates_clifford}.

\begin{figure}[h]
    \centering
    \includegraphics[width=0.75\linewidth]{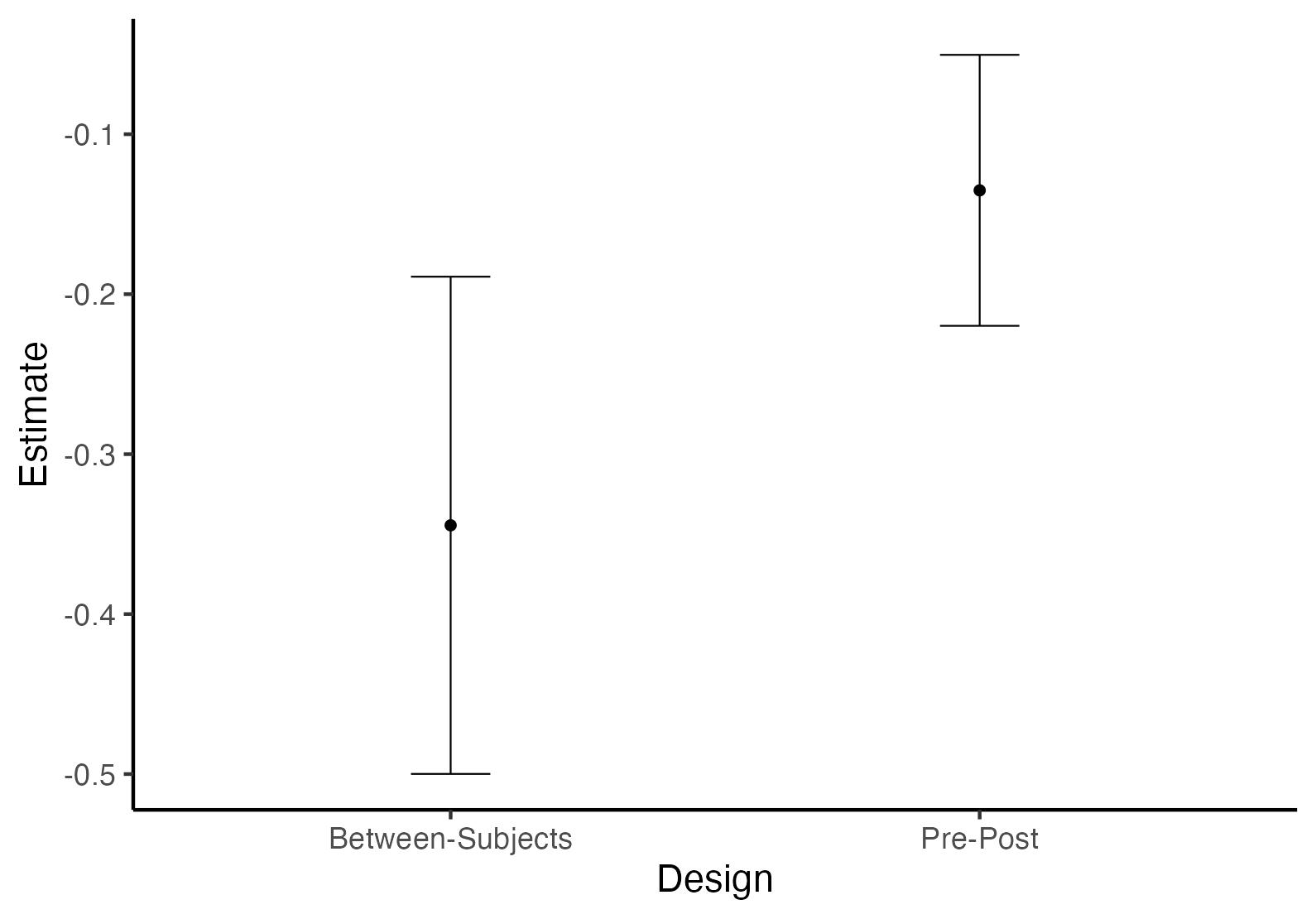}
    \caption{Experiment 2 in \cite{09_Clifford}}
    \label{fig:difference-in-estimates_clifford}
\end{figure}

This example demonstrates how experimental design choices can markedly affect inference. Specifically, it suggests that the estimated treatment effect varies with the design, potentially due to carryover effects. Formally, we observe that:
\[
 \tau^{\text{between-subjects}} = \mathbb{E}[Y_i(1) - Y_i(0)] \neq \tau^{\text{pre-post}} = \mathbb{E}[Y_i(0,1) - Y_i(0,0)]
\]
Thus, the notion of \emph{sequential exchangeability} provides a framework to examine the conditions under which these two estimands may be equivalent.

\subsection{Sequential Exchangeability}

In the \cwsd{} design, the sequence in which participants receive the treatment or control is randomized such that $Z_{i,t_1}$ determines $Z_{i,t_2}$; specifically, $Z_{i,t_1} = 1 - Z_{i,t_2}$. In contrast, the pre-post design employed in Experiment 2 in \cite{09_Clifford} randomizes participants' treatment status at time $t_2$, with all participants receiving the control in time $t_1$.

Moreover, other sequential experimental designs randomize treatment assignment in both time periods, as will be discussed further in Section \ref{section: alternative_designs}. In order to generalize, we can redefine Equation \ref{equation: t2_identity_cwsd} as follows:
\begin{equation}
\label{equation: t2_identity} 
    Y_{i,t_2}^{\text{obs}} = \sum_{z_{i,t_1} \in \{0,1\}} \sum_{z_{i,t_2} \in \{0,1\}} Y_{i,t_2}(z_{i,t_1}, z_{i,t_2})\,\mathbbm{1}\{Z_{i,t_1} = z_{i,t_1}\}\,\mathbbm{1}\{Z_{i,t_2} = z_{i,t_2}\}.
\end{equation}

\noindent
We now formally define \emph{sequential exchangeability} as follows:

\begin{assumption}(Sequential Exchangeability)
\label{assumption: sequential_exchangeability}
 A sequential experimental design with two time periods and two treatment conditions satisfies \emph{sequential exchangeability} if the following conditions are met:
    \begin{assumptionnum}
        \item \label{assumption: sequential_exchangeability: controlled_carryover} 
        \textbf{Controlled carryover effects: } 
        \[
        Y_{i,t_2}(z_{i,t_1}, z_{i,t_2}) = Y_{i,t_2}(z_{i,t_2}) \mid Z_{i,t_1} = z_{i,t_1}, Z_{i,t_2} = z_{i,t_2}, X_{i,t_2}, \quad \forall z_{i,t_1}, z_{i,t_2} \in \{0,1\}
        \] 

        \item \label{assumption: sequential_exchangeability: ignorability_t2} 
        \textbf{Ignorability at time $t_2$:} 
        \[
        Y_{i,t_2}(z_{i,t_2}) \indep  Z_{i,t_2} \mid X_{i,t_2}, \quad  \forall z_{i,t_2} \in \{0,1\}.
        \]
        \item \label{assumption: sequential_exchangeability: parallel_trend} 
        \textbf{Parallel trend assumption:} 
        \[
        Y_{i,t_2}(z_{i,t_2}) = Y_{i,t_1}(z_{i,t_2}) + c, \quad \forall z_{i,t_2} \in \{0,1\},\ c \in \mathbb{R}.
        \]
    \end{assumptionnum}
    These conditions are considered in addition to SUTVA and the overlap assumptions stipulated by the Rubin Causal Model.
\end{assumption}

\noindent 
The generalization to $n$ time periods and $n$ conditions can be found in the Appendix.

The primary distinction between \emph{sequential exchangeability} and the assumptions required to identify the average treatment effect (ATE) in \cwsd{} lies in Assumption \ref{assumption: sequential_exchangeability: controlled_carryover}. Intuitively, if carryover effects can be appropriately controlled—such that the potential outcome at $t_2$ corresponds to that which would be observed if the participant were recruited at $t_2$ without any prior intervention. We want to note however that this is a \emph{strong} assumption that can be easily violated. We explore some of the practical challenges in controlling for the direct carryover effects in Section \ref{sec: control_carryover_eff_challenges}.

Evaluating Experiment 2 in \cite{09_Clifford} with respect to Assumption \ref{assumption: sequential_exchangeability}, it is plausible that Assumptions \ref{assumption: sequential_exchangeability: ignorability_t2} and \ref{assumption: sequential_exchangeability: parallel_trend} hold, given that participants were randomized at $t_2$ and the parallel trend assumption seem plausible. However, due to the nature of the experiment, administering the question at $t_1$ may introduce carryover effects that are challenging to control for. The mechanism by which treatment at $t_1$ influences the outcome at $t_2$ is ambiguous, thereby potentially violating Assumption \ref{assumption: sequential_exchangeability: controlled_carryover}.

\subsection{Alternative Sequential Experimental Designs}
\label{section: alternative_designs}

Based on the sequential exchangeability assumptions, we propose several alternative experimental designs to \cwsd{} that may be more appropriate in contexts where the underlying assumptions are more plausible. Although these designs have been previously employed for diverse purposes as noted in the introduction, our focus here is on their utility for estimating the average treatment effect (ATE).

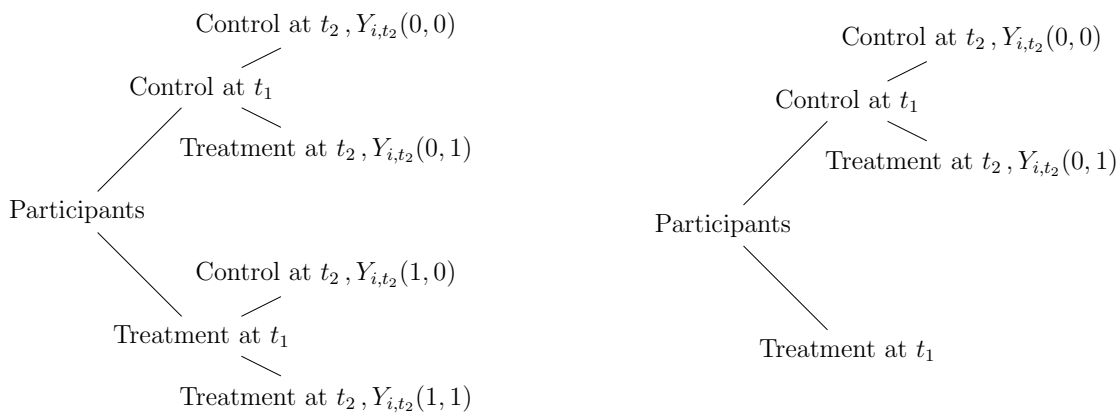
\begin{figure}[h]
    \centering
    \begin{subfigure}[b]{0.45\textwidth}
    \resizebox{\textwidth}{!}{%
    \fbox{
    \begin{tikzpicture}[
  grow=right,
  level 1/.style={sibling distance=4cm, level distance=2cm},
  level 2/.style={sibling distance=2cm, level distance=2cm}]
    \node {Participants}
    child { node {Treatment at \(t_1\)}
    child { node {Treatment at \(t_2 \,, Y_{i,t_2}(1,1)\)} }
    child { node {Control at \(t_2 \,, Y_{i,t_2}(1,0)\)} }
    }
    child { node {Control at \(t_1\)}
    child { node {Treatment at \(t_2 \,, Y_{i,t_2}(0,1)\)} }
    child { node {Control at \(t_2 \,, Y_{i,t_2}(0,0)\)} }
  };
\end{tikzpicture}}
}
\subcaption{Sequential Randomization}
\label{fig:alt_designs_both}
\end{subfigure}
\hspace{10mm}
\begin{subfigure}[b]{0.45\textwidth}
\resizebox{\textwidth}{!}{%
\fbox{
    \begin{tikzpicture}[
  grow=right,
  level 1/.style={sibling distance=4cm, level distance=2cm},
  level 2/.style={sibling distance=2cm, level distance=2cm}],
    \node {Participants}
    child { node {Treatment at \(t_1\)}
    child { node {} edge from parent[draw=none] }
    child { node {} edge from parent[draw=none] }
    }
    child { node {Control at \(t_1\)}
    child { node {Treatment at \(t_2 \,, Y_{i,t_2}(0,1)\)} }
    child { node {Control at \(t_2 \,, Y_{i,t_2}(0,0)\)} }
  };
\end{tikzpicture}}
}
\subcaption{Selective Sequential Randomization}
\label{fig:alt_designs_selective}
\end{subfigure}
    \caption{Alternative Experimental Designs to \Cwsd{}}
    \label{fig:alt_designs}
\end{figure}

A notable advantage of the sequential randomization design depicted in Figure \ref{fig:alt_designs_both} is its ability to incorporate \(Z_{i,t_1}\) as a control variable without introducing perfect collinearity. In traditional \cwsd{} settings, the deterministic relationship \(Z_{i,t_1} = 1 - Z_{i,t_2}\) renders it infeasible to include both treatment indicators in an OLS regression model. By contrast, the sequential randomization design circumvents this issue, thereby allowing more robust estimation strategies. Nonetheless, a key challenge remains in accurately characterizing the nature of carryover effects; the risk of misspecification looms large if these effects are not thoroughly understood.

The most significant benefit of the proposed experimental designs is the credibility of the ignorability assumption at time \(t_2\) (Assumption \ref{assumption: sequential_exchangeability: ignorability_t2}). When randomization is implemented independently in the second stage, the potential for unobserved confounding—particularly through indirect carryover effects—is substantially reduced. This is especially valuable since  experimental designs are often applied to estimate the average treatment effect, without assuming how potential covariates might confound the estimate. 

Sequential experimental designs can and should be carefully tailored to the specific context of an experiment. For instance, consider the selective sequential randomization design illustrated in Figure \ref{fig:alt_designs_selective}. If carryover effects are negligible in one of the conditions—such as when the control involves a placebo (e.g., a sugar pill assumed to have no effects of the outcome of interest)—then it may be reasonable to relax Assumption \ref{assumption: sequential_exchangeability: controlled_carryover}, where the assumption is true for \(z_{i,t_1} = 0\) but not \(z_{i,t_1} = 1\). In this scenario, the selective design could yield reliable ATE estimates. However, if the control represents a current standard treatment and the experimental arm involves a novel intervention, where both conditions may have significant carryover effects, relaxing the assumption would be inappropriate, as the differential impact of prior exposure could bias the estimation of ATE. This emphasis on accounting for the nature of the experiment and treatment, and on adapting the randomization of treatment assignment for causal identification, echoes the rationale behind dynamic treatment regimes and Sequential Multiple Assignment Randomized Trial (SMART) designs, which similarly account for sequential assignments to estimate causal effects under minimal modeling assumptions \citep{25_SMARTDesign}.

For completeness, we state the following corollary linking sequential exchangeability to the validity of the proposed experimental designs.

\begin{cor}
\label{cor: two-period-two-arm}
  In \emph{any} two‐period, two‐arm sequential experimental design, if \emph{sequential exchangeability}
  holds, then
  \[
    \tau_{t_2}^{\mathrm{seq}} \;=\; \tau.
  \]
\end{cor} 

The proof proceeds analogously to that of Theorem \ref{theorem: identify_cwsd}.

\begin{remark}
  Equation \ref{equation: t2_identity} in fact provides the most general two‐period, two‐arm sequential‐outcome set‐up that includes counterbalanced within‐subjects designa, the pre–post designs, sequential randomization designs,  selective sequential randomization designs, if \emph{sequential exchangeability}
  holds, then
  \[
    \tau_{t_2}^{\mathrm{seq}} \;=\; \tau.
  \] Hence whenever the observed second‐period outcome in your design can be written in this form, the proof of Corollary \ref{cor: two-period-two-arm} goes through verbatim under the sequential exchangeability assumption. 
\end{remark}

\subsection{Controlling for Carryover Effects in Practice}
\label{sec: control_carryover_eff_challenges}

Carryover effects from the treatment at time $t_1$ can be viewed as a form of omitted variable bias—essentially, a `covariate' that needs to be properly accounted for. As with any covariate, misspecifying its functional form can lead to biased estimates. If carryover effects are linearly separable, then a fixed-effects model that includes treatment status at $t_1$ as a control is typically sufficient. However, when carryover effects interact with the subsequent treatment at $t_2$, estimation becomes more complex.

Consider two illustrative models where carryover effects cannot be fully addressed by simple fixed-effects approaches. If the carryover effects from $t_1$ interacts outcomes at $t_2$ in an additive way, we can model the outcome as:

\begin{equation}
\label{eq: interaction_carryover}
Y_{i,t_2} 
= \beta_0 + \beta_1 Z_{i,t_2} 
+ \gamma (Z_{i,t_1} \cdot Z_{i,t_2}) 
+ \beta_2 X_{i, t_2} 
+ \epsilon_{i, t_2}, 
\quad \epsilon_{i, t_2} \sim N(0,\sigma^2).
\end{equation}

Here, the effect of the second treatment $Z_{i,t_2}$ depends additively on whether the first treatment $Z_{i,t_1}$ was received, through the interaction term $\gamma (Z_{i,t_1} \cdot Z_{i,t_2})$.

If, instead, the prior treatment amplifies or dampens the effect of the later treatment in a nonlinear way—as in the following model:
\begin{equation}
\label{eq: compounding_carryover}
Y_{i,t_2} 
= \beta_0 
+ \beta_1^{1 + \gamma Z_{i,t_1}} Z_{i,t_2} 
+ \beta_2 X_{i,t_2} 
+ \epsilon_{i, t_2}, 
\quad \epsilon_{i, t_2} \sim N(0,\sigma^2).
\end{equation}
In this formulation, the strength of the effect of \( Z_{i,t_2} \) is modulated by prior exposure \( Z_{i,t_1} \) through an exponentiated function. When \( \gamma > 0 \), prior treatment increases the effective coefficient on \( Z_{i,t_2} \), amplifying its influence on the outcome. Conversely, when \( \gamma < 0 \), the coefficient is dampened.

\begin{figure}[h]
    \centering
    \includegraphics[width=0.95\linewidth]{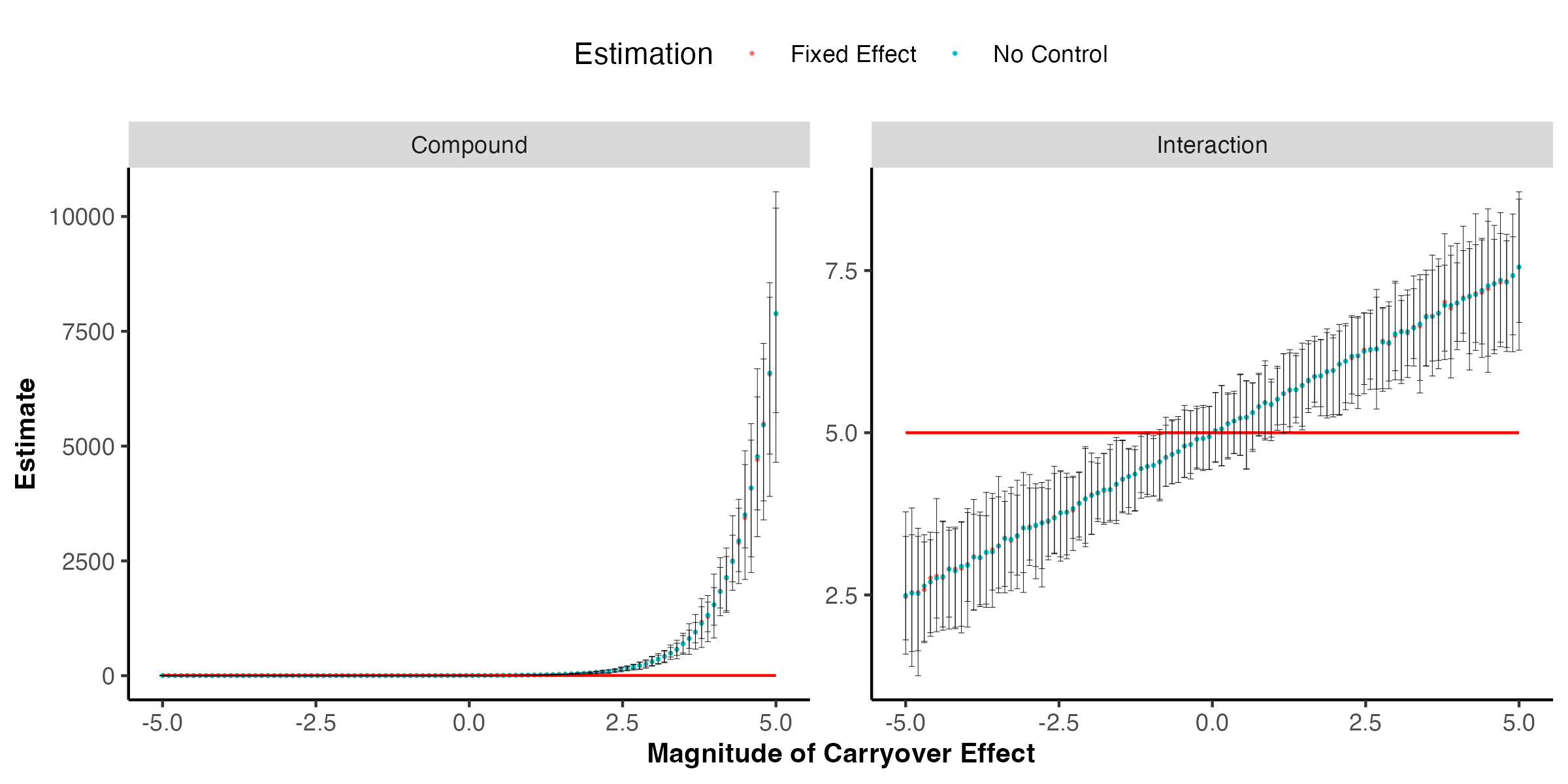}
    \caption{Simulated estimates of different structures of carryover effects of treatment at $t_1$ on the treatment at $t_2$. On the left, we have the compounding effect (Equation \ref{eq: compounding_carryover}), whereas on the right, we have the simple interaction effect (Equation \ref{eq: interaction_carryover}) between $Z_{t_1}$ and $Z_{t_2}$. The true average treatment effect, $\tau_{\text{ATE}}$, is set at 5 (red horizontal line).}
    \label{fig:simulation_control_carryover}
\end{figure}

Figure \ref{fig:simulation_control_carryover} shows the bias of the estimated treatment effect, $\beta_1$, at $t_2$ varying the magnitude of the carryover effect, $\gamma$, as specified in Model 1 and Model 2. These two examples show how misspecification or ignoring the correct functional form of carryover can bias estimates. In Model 1, adding a simple interaction term can often account for such dependence, whereas in Model 2 requires a different approach to properly estimate the true effect of the second treatment.

\cite{23_AnalysisOfCrossoverTrial} provides an analysis of estimating \cwsd{} (referred to as 2×2 crossover trials) in the case of no carryover effects, as well as additive and interactive carryover effects, using OLS and REML estimators. However, as noted, the functional form of the carryover effect can be difficult to discern, and misspecification can lead to biased estimates.
\section{Strategies for Addressing Identification Challenges}
\label{sec: what_can_we_do}

Although counterbalanced within‐subjects designs present some identification challenges, there are practical strategies available to researchers to address these problems. This section outlines five approaches—Fisher's Exact Test, heuristic checks, the use of a washout period, covariate adjustment, and using Sequential Randomization as described in Section \ref{sec: seq_exchangeability_alt_design}—that can help detect violations of the assumptions or satisfy them needed for consistent estimation of the causal effect.

\subsection{Fisher's Exact Test}

Fisher's Exact Test is a nonparametric method widely used to assess the independence between two categorical variables, making it particularly useful in small sample settings. In a \cwsd{}, if there is concern that assumption violations may be contaminating the analysis, one strategy is to discard data from $t_2$ and focus exclusively on $t_1$. By doing so, researchers can use Fisher's Exact Test at $t_1$ to determine whether the treatment assignment is independent of the observed outcomes.

However, this approach does not address the issue of limited overlap and lack of common support, which may arise with the sample sizes typically used in \cwsd{}. Therefore, while it is a viable option, the experimenter may still need to ensure covariate balance at $t_1$, and its applicability may be limited in certain contexts.

\subsection{Heuristics Checks}

A straightforward diagnostic is to compare the estimated treatment effect in the first period, \(\hat{\tau}_{t_1}\), to the estimated effect in the second period, \(\hat{\tau}_{t_2}\). A marked discrepancy between these estimates may indicate violations of the no‐carryover or no‐time‐drift assumptions (analogous to pre‐trend checks in a difference‐in‐differences setting).

When large differences emerge, they serve as a clear warning sign that carryover or other time‐based confounders may be at play.  However, consistency between \(\hat{\tau}_{t_1}\) and \(\hat{\tau}_{t_2}\) does not guarantee the validity of the requisite assumptions: random sampling variability could mask carryover effects, while strongly offsetting biases might also produce superficially similar estimates. Inconsistencies between the estimates also do not imply that the assumptions are violated as differences might arise due to sampling variability. As such, this check is neither necessary (i.e., estimates may differ for benign reasons) nor sufficient (i.e., matching estimates could still reflect compensating biases).

Despite these limitations, heuristic checks remain a useful first step in diagnosing possible threats to identification and can guide further analyses. However, given that \cwsd{} is often employed due to a small sample size, in practice, it can be difficult to determine that difference in estimates is driven by the violation of the assumptions, or sampling variation, or both.

\subsection{Washout Period}

A more targeted design‐based remedy is to incorporate a \textit{washout period} between treatments. This approach is particularly relevant when the treatment is expected to have direct physiological or psychological effects that persist over time, but decay with a sufficient gap between $t_1$ and $t_2$. By lengthening the gap between \(t_1\) and \(t_2\), the influence of the initial treatment on the second‐period outcome may subside, thereby reducing or eliminating any carryover effects.  

 Researchers should draw on subject‐matter knowledge (e.g., pharmacokinetics, psychological adaptation timelines) to determine the length of time needed for treatment effects to dissipate. Even with an adequate washout period, however, confounding due to the effect of treatment on the covariates remains a concern. If, for instance, the initial treatment induces behavioral or physiological changes that persist, mere passage of time may not fully remove indirect carryover pathways. Thus, the potential for lingering confounders should still be carefully addressed (e.g., by measuring and adjusting for relevant characteristics at \(t_2\)).

\subsection{Controlling for Changes in Covariates}

\begin{figure}[h]
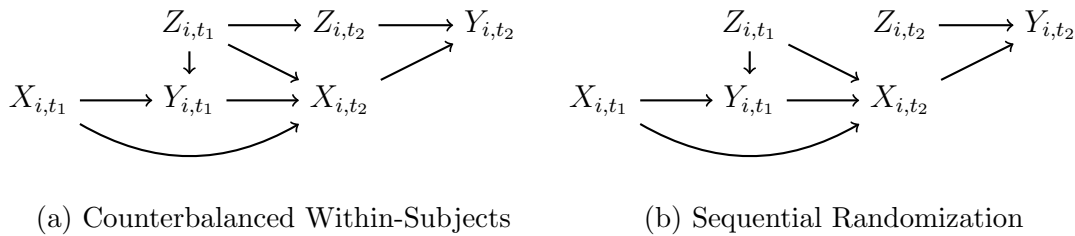

    \centering
    \begin{subfigure}[b]{0.45\textwidth}
    \tikz{
        \node (space) at (0,1) {};
        \node (space) at (0,-2) {};
        \node (Zt1) at (2,0) {$Z_{i,t_1}$};
        \node (Xt1) at (0,-1) {$X_{i,t_1}$};
        \node (Yt1) at (2,-1) {$Y_{i,t_1}$};
        
        \node (Zt2) at (4,0) {$Z_{i,t_2}$};
        \node (Xt2) at (4,-1) {$X_{i,t_2}$};
        \node (Yt2) at (6,0) {$Y_{i,t_2}$};
    
        \path[->, thick] (Xt1) edge (Yt1);
        \path[->, thick] (Zt1) edge (Yt1);
        \path[->, thick] (Zt1) edge (Zt2);
        \path[->, thick] (Yt1) edge (Xt2);
        \path[->, thick] (Zt2) edge (Yt2);
        \path[->, thick] (Xt2) edge (Yt2);
        \path[->, thick] (Zt1) edge (Xt2);
        \path[->, thick] (Xt1) edge[out=-30, in=-150] (Xt2);
    }
    \subcaption{Counterbalanced Within-Subjects}
    \label{fig: cwsd_dag_no_dir}
    \end{subfigure} \vspace{1mm}
    \begin{subfigure}[b]{0.45\textwidth}
    \tikz{
        \node (space) at (0,1) {};
        \node (space) at (0,-2) {};
        \node (Zt1) at (2,0) {$Z_{i,t_1}$};
        \node (Xt1) at (0,-1) {$X_{i,t_1}$};
        \node (Yt1) at (2,-1) {$Y_{i,t_1}$};
        
        \node (Zt2) at (4,0) {$Z_{i,t_2}$};
        \node (Xt2) at (4,-1) {$X_{i,t_2}$};
        \node (Yt2) at (6,0) {$Y_{i,t_2}$};
    
        \path[->, thick] (Xt1) edge (Yt1);
        \path[->, thick] (Zt1) edge (Yt1);
        \path[->, thick] (Yt1) edge (Xt2);
        \path[->, thick] (Zt2) edge (Yt2);
        \path[->, thick] (Xt2) edge (Yt2);
        \path[->, thick] (Zt1) edge (Xt2);
        \path[->, thick] (Xt1) edge[out=-30, in=-150] (Xt2);
    }
    \subcaption{Sequential Randomization}
    \label{fig: sr_dag_no_dir}
    \end{subfigure}
    \caption{Assuming that all carryover effects are mediated by the covariates at the second time period, failing to control for $X_{i,t_2}$ will lead to biased estimates in a \cwsd{}. However, this is not true for Sequential Randomization. Note that $X_{i,t_1}, X_{i,t_2}$ includes both observed and unobserved covariates for concision.}
    \label{fig:CWSD_DAG_Without_Direct}
\end{figure}

If the carryover effect is entirely mediated by observable covariates at the second time period, \(X_{i,t_2}\) (i.e., if all carryover effects operate indirectly through these covariates), then adjusting for \(X_{i,t_2}\) fully accounts for any such indirect influence. Formally, controlling for the covariates at $t_2$ satisfies Assumption \ref{assumption: sequential_exchangeability: ignorability_t2}.

To illustrate with the Ozempic experiment, suppose that any direct carryover effects of the treatment on blood sugar levels at \(t_1\) become negligible following an appropriate washout period; then the assumption of no direct carryover effects becomes more credible. However, if participants assigned to Ozempic continue to maintain elevated physical activity relative to those in the control group, a \cwsd{} will violate Assumption \ref{assumption: sequential_exchangeability: ignorability_t2}.

To account for this in a \cwsd{}, we can adapt methods from observational causal inference by adjusting for $X_{i,t_2}$ using regression models or propensity score weighting. These approaches block indirect pathways through which prior treatment might influence the outcome. In contrast, a simple time dummy only captures average period effects and fails to account for treatment–covariate interactions. Nevertheless, while such adjustments are theoretically viable in a \cwsd{}, they may be difficult to implement in practice, and model misspecification can introduce bias. Therefore, even when direct carryover effects seem unlikely or negligible, we recommend using sequential randomization as described in Section \ref{section: alternative_designs}.

To demonstrate this, we simulated a \cwsd{} where participants were randomized to sequences, and varied the effect of the effect of $Z_{t_1}$ on $X_{t_2}$ from –10 to 10. We estimated the treatment effect using four methods: a correctly specified model, a propensity score method, a fixed time effect model, and an unadjusted model. 

As shown in Figure \ref{fig:simulation_cwsd}, only the correctly specified and propensity score models consistently recovered the true treatment effect. The fixed effect and unadjusted models produced biased estimates, with bias increasing alongside the treatment’s effect on the covariates. We repeated the simulation using sequential randomization, which yielded unbiased estimates regardless of model (mis)specification.

\begin{figure}[h]
    \centering
    \includegraphics[width=0.95\linewidth]{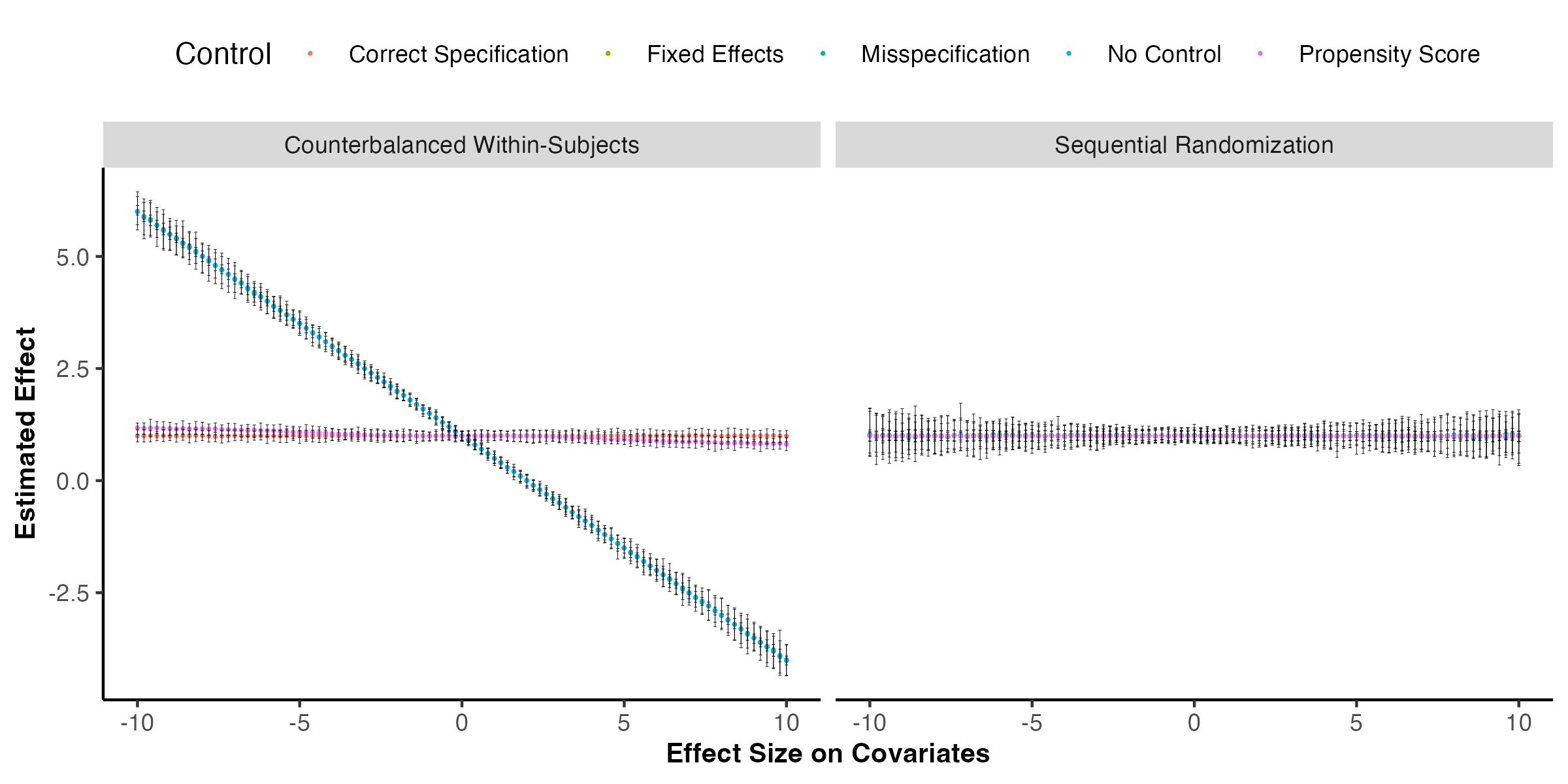}
    \caption{Simulated estimates of the treatment effect \(\tau = 1\) obtained using different estimation methods. In the \cwsd{} simulation, both the Correct Specification and Propensity Score approaches successfully recovered the unbiased treatment effect, while the Fixed Effects and No Control methods introduced bias. The degree of bias in the Fixed Effects and No Control approaches increased as the effect size on the covariates grew. Nonetheless, a sequential randomization design consistently recovered the average treatment effect (ATE) regardless of the specification.}
    \label{fig:simulation_cwsd}
\end{figure}

Although adjusting for covariates at $t_2$ can help block indirect carryover pathways, care must be taken not to condition on previous treatment outcomes such as \( Y_{i,t_1} \). In particular, consider the path \[ Z_{i,t_2} \leftarrow Z_{i,t_1} \rightarrow Y_{i,t_1} \leftarrow X_{i,t_1} \rightarrow X_{i,t_2} \rightarrow Y_{i,t_2} \] $Y_{i, t_1}$ is a collider and conditioning on it opens a non-causal backdoor path from \( Z_{i,t_1} \) to \( Y_{i,t_2} \). This violates the backdoor criterion and introduces what is known as collider bias or \textit{M}-bias \citep{21_M_Bias, 22_good_and_bad_crtl}. As a result, estimates of the causal effect may become biased. Thus, in a \cwsd{}, treatment outcomes at $t_1$ should not be conditioned on when estimating causal effects of \( Z_{i,t_2} \) on $Y_{i,t_2}$.
\section{Discussion and Conclusion}
\label{sec: conclusion}

Our analysis of counterbalanced within-subjects designs shows the potential challenges in using this methodology for causal inference. While this design can offer advantages—such as increased statistical power, more precise individual-level comparisons, and efficient data collection—it also introduces complications that can bias treatment effect estimates if not carefully accounted for. The primary concern stems from carryover effects, which may persist asymmetrically across treatment conditions, and violations of stability over treatment and time, which can confound treatment estimates across sequential periods. 

Through a formal treatment of these issues within the potential outcomes framework, we demonstrate that counterbalancing does not inherently eliminate biases introduced by differential carryover effects. Our findings suggest that, although counterbalancing is often assumed to mitigate such issues, it relies on the strong and unverifiable assumption that carryover effects are symmetric and cancel out. When this assumption does not hold treatment effect estimates may be substantially biased.

Furthermore, issues that arise in \cwsd{} also appear in other sequential experimental designs where the estimand of interest is the ATE. We formalize the notion of \emph{sequential exchangeability}, which applies to alternative experimental designs to \cwsd{} where the assumptions are more plausibly met. However, we caution that despite the improvements, the difficulty of controlling for carryover effects can substantially hinder the estimation of the ATE and is not a perfect substitute for \bsd{}.

To address these concerns, we propose several practical strategies that researchers can employ when implementing counterbalanced within-subjects designs. First, heuristic checks comparing treatment effects across different time periods can serve as an initial diagnostic tool to detect potential violations of key assumptions. Second, incorporating a washout period between conditions can reduce the impact of lingering treatment effects, particularly in studies involving physiological or psychological interventions. Finally, adjusting for covariate changes between periods can help correct for biases arising from time-dependent confounders that may be influenced by prior treatment exposure, assuming that all carryover effects can be explained indirectly via the covariates.

Despite its limitations, counterbalanced within-subjects design can be a valuable tool for experimental research  participant recruitment is constrained. However, researchers should assess whether \cwsd{} is the most appropriate design choice for their specific study context, or whether alternative methodologies—such as between-subjects designs or sequential randomization—might provide more reliable estimates of causal effects.

\bibliographystyle{chicago}

\bibliography{cite}

\newpage

\section*{Appendix}

\subsection*{Proofs}

\begin{retheorem}[Restatement of Theorem \ref{theorem: identify_bsd}]
    Under the standard assumptions of the Rubin Causal Model, ATE is identifiable in a between-subjects design.
\end{retheorem}

\begin{proof}
    We denote the observable outcomes as $Y_i^{\text{obs}}$ for clarity. We take the difference of the expected observed outcomes between those who received treatment $Z_i = 1$ and those those who did not $Z_i = 0$, conditioning on their covariates
    
    \[\mathbb{E}[Y_{i}^{\text{obs}} \mid Z_i = 1, X_i] - \mathbb{E}[Y_{i}^{\text{obs}} \mid Z_i = 0, X_i]\] 
 
    We can write the following strings of equalities
    \[
        \begin{aligned}
            &\quad \,\,\mathbb{E}[Y_{i}^{\text{obs}} \mid Z_i = 1, X_i] - \mathbb{E}[Y_{i}^{\text{obs}} \mid Z_i = 0, X_i] \\
            &= \mathbb{E}[Y_i(1)Z_i + Y_i(0)(1 - Z_i) \mid Z_i = 1, X_i] - \mathbb{E}[Y_i(1)Z_i + Y_i(0)(1 - Z_i) \mid Z_i = 0, X_i] \\
            &= \mathbb{E}[Y_i(1) \mid Z_i = 1, X_i] - \mathbb{E}[Y_i(0) \mid Z_i = 0, X_i] \\
            &= \mathbb{E}[Y_i(1) \mid X_i] - \mathbb{E}[Y_i(0) \mid X_i] \quad \textit{(by ignorability from randomization)} \\ 
            &= \mathbb{E}[Y_i(1) - Y_i(0) \mid X_i] 
        \end{aligned}
    \] which is the Conditional Average Treatment Effect. With the Conditional Average Treatment Effect, we can apply Adam's Law to obtain the Average Treatment Effect.

    \[\begin{aligned}
        \mathbb{E}_{X_i}\mathbb{E}[Y_i(1) - Y_i(0) \mid X_i] =  \mathbb{E}[Y_i(1) - Y_i(0)]  = \tau
    \end{aligned}\]
\end{proof}

\begin{reprop}[Restatement of Proposition \ref{prop: FWL}]
The OLS estimate from Equation \eqref{eq: OLS_estimate}, \(\hat{\alpha}_1 = \hat{\tau}^{\text{CWSD}}\), can be expressed as a convex combination of the period-specific estimates:
\[
\hat{\tau}^{\text{CWSD}} = q\,\hat{\tau}_{t_1} + (1-q)\,\hat{\tau}_{t_2},
\]
with \(q \in [0,1]\).
\end{reprop}

\begin{proof}
For notational simplicity, define
\[
W_{i,t} = \begin{pmatrix} \mathbbm{1}\{t = t_1\} \\[1mm] X_{i,t} \end{pmatrix}.
\]
Then, the regression in Equation \eqref{eq: OLS_estimate} can be rewritten as
\[
Y_{i,t} = \alpha_0 + \alpha_1 Z_{i,t} + W_{i,t}'\beta + \varepsilon_{i,t},
\]
where
\[
\beta = \begin{pmatrix} \alpha_2 \\ \alpha_3 \end{pmatrix}.
\]

To isolate the effect of \(Z_{i,t}\), we first partial out the influence of \(W_{i,t}\) by estimating the auxiliary regressions:
\[
Y_{i,t} = \hat{\gamma}_0 + W_{i,t}'\hat{\gamma}_1 + r^Y_{i,t} \quad \text{and} \quad Z_{i,t} = \hat{\delta}_0 + W_{i,t}'\hat{\delta}_1 + r^Z_{i,t},
\]
where \(r^Y_{i,t}\) and \(r^Z_{i,t}\) are the residuals. By the Frisch–Waugh–Lovell theorem, the OLS estimator for \(\alpha_1\) is given by
\[
\begin{aligned}
    \hat{\alpha}_1 &= \frac{\sum_{t\in\{t_1,t_2\}} \sum_{i=1}^n (r^Y_{i,t} - \bar{r}^Y_{i,t})\,(r^Z_{i,t} - \bar{r}^Z_{i,t})}{\sum_{t\in\{t_1,t_2\}} \sum_{i=1}^n (r^Z_{i,t} - \bar{r}^Z_{i,t})^2} \\
&= \frac{\sum_{t\in\{t_1,t_2\}} \sum_{i=1}^n r^Y_{i,t}\,r^Z_{i,t}}{\sum_{t\in\{t_1,t_2\}} \sum_{i=1}^n (r^Z_{i,t})^2} \quad since \,\,\bar{r}^Z_{i,t} = \bar{r}^Y_{i,t} = 0
\end{aligned}
\]
We can decompose the numerator by time period:
\[
\hat{\alpha}_1 = \frac{\sum_{i=1}^n r^Y_{i,t_1}\,r^Z_{i,t_1}}{\sum_{t\in\{t_1,t_2\}} \sum_{i=1}^n (r^Z_{i,t})^2} + \frac{\sum_{i=1}^n r^Y_{i,t_2}\,r^Z_{i,t_2}}{\sum_{t\in\{t_1,t_2\}} \sum_{i=1}^n (r^Z_{i,t})^2}.
\]

Focusing on the \(t_1\) component, note that
\[
\frac{\sum_{i=1}^n r^Y_{i,t_1}\,r^Z_{i,t_1}}{\sum_{t\in\{t_1,t_2\}} \sum_{i=1}^n (r^Z_{i,t})^2} = \left(\frac{\sum_{i=1}^n (r^Z_{i,t_1})^2}{\sum_{t\in\{t_1,t_2\}} \sum_{i=1}^n (r^Z_{i,t})^2}\right) \left(\frac{\sum_{i=1}^n r^Y_{i,t_1}\,r^Z_{i,t_1}}{\sum_{i=1}^n (r^Z_{i,t_1})^2}\right) = q\,\hat{\tau}_{t_1},
\]
where
\[
q = \frac{\sum_{i=1}^n (r^Z_{i,t_1})^2}{\sum_{t\in\{t_1,t_2\}} \sum_{i=1}^n (r^Z_{i,t})^2} \in [0,1],
\]
and \(\hat{\tau}_{t_1}\) is the treatment effect estimated from the \(t_1\) period alone. An analogous argument shows that the \(t_2\) component can be written as
\[
\frac{\sum_{i=1}^n r^Y_{i,t_2}\,r^Z_{i,t_2}}{\sum_{t\in\{t_1,t_2\}} \sum_{i=1}^n (r^Z_{i,t})^2} = (1-q)\,\hat{\tau}_{t_2}.
\]
Thus, we conclude that
\[
\hat{\tau}^{\text{CWSD}} = \hat{\alpha}_1 = q\,\hat{\tau}_{t_1} + (1-q)\,\hat{\tau}_{t_2}.
\]
\end{proof}

\begin{retheorem}[Restatement of Theorem \ref{theorem: identify_cwsd}]
Under the standard Rubin Causal Model assumptions (SUTVA and overlap), the Average Treatment Effect (ATE) is identifiable in a \cwsd{} at $t_2$ if the following assumptions hold:
\begin{enumerate}
    \item \textbf{Simple Direct Carryover Effects:} As defined in Assumption \ref{assumption:simple_carryover_effects}.
    \item \textbf{Ignorability at Time \(t_2\):} 
    \[
    Y_{i,t_2}(z_{i,t_2}) \indep Z_{i,t_2} \mid X_{i,t_2}, \quad \forall z_{i,t_2} \in \{0,1\}.
    \]
    \item \textbf{Parallel Trends:} 
    \[
    Y_{i,t_2}(z_{i,t_2}) = Y_{i,t_1}(z_{i,t_2}) + c, \quad \forall z_{i,t_2} \in \{0,1\}, \quad c \in \mathbb{R}.
    \]
\end{enumerate}
\end{retheorem}

\begin{proof}
    
    In a counterbalanced within-subjects design, we only observe \[\mathbb{E}[Y_{i,t_2}^{\text{obs}} \mid Z_{i,t_1} = 1, Z_{i,t_2} = 0, X_{i,t_2}] \quad \text{and} \quad \mathbb{E}[Y_{i,t_2}^{\text{obs}} \mid Z_{i,t_1} = 0, Z_{i,t_2} = 1, X_{i,t_2}]\] 

    Taking the difference for those that receive treatment at time $t_2$ and the control group, yields
    \[\begin{aligned}
        & \quad \mathbb{E}[Y_{i,t_2}^{\text{obs}} \mid Z_{i,t_1} = 0, Z_{i,t_2} = 1, X_{i,t_2}] - \mathbb{E}[Y_{i,t_2}^{\text{obs}} \mid Z_{i,t_1} = 1, Z_{i,t_2} = 0, X_{i,t_2}] \\
        &= \mathbb{E}[Y_{i,t_2}(0,1) \mid Z_{i,t_1} = 0, Z_{i,t_2} = 1, X_{i,t_2}] - \mathbb{E}[Y_{i,t_2}(1,0) \mid Z_{i,t_1} = 1, Z_{i,t_2} = 0, X_{i,t_2}] \quad \textit{(by identity \ref{equation: t2_identity_cwsd})} \\
        &= \mathbb{E}[Y_{i,t_2}(1) \mid Z_{i,t_2} = 1, X_{i,t_2}] - C - \mathbb{E}[Y_{i,t_2}(0) \mid  Z_{i,t_2} = 0, X_{i,t_2}] + C \quad \textit{(by Assumption \ref{assumption:simple_carryover_effects}}) \\
        &= \mathbb{E}[Y_{i,t_2}(1) \mid Z_{i,t_2} = 1, X_{i,t_2}] - \mathbb{E}[Y_{i,t_2}(0) \mid  Z_{i,t_2} = 0, X_{i,t_2}]  \\
        &= \mathbb{E}[Y_{i,t_2}(1) \mid X_{i,t_2}] - \mathbb{E}[Y_{i,t_2}(0) \mid  X_{i,t_2}] \quad \textit{(by ignorability at $t_2$)} \\
        &= \mathbb{E}[Y_{i,t_2}(1) - Y_{i,t_2}(0) \mid X_{i,t_2}] \quad \textit{(by linearity of expectations)} 
    \end{aligned}
    \] 

    With the Conditional Average Treatment Effect, we can apply Adam's Law to obtain the Average Treatment Effect.

    \[\begin{aligned}
        \mathbb{E}_{X_{i,t_2}}\mathbb{E}[Y_{i,t_2}(1) - Y_{i,t_2}(0) \mid X_{i,t_2}] 
        &= \mathbb{E}[Y_{i,t_2}(1) - Y_{i,t_2}(0)] \\
        &= \mathbb{E}[Y_{i,t_1}(1) + c - Y_{i,t_1}(0) - c\,] \quad \textit{(by parallel trends assumption)} \\
        &= \mathbb{E}[Y_{i,t_1}(1) - Y_{i,t_1}(0)]\\
        &= \tau
    \end{aligned}\]
\end{proof}

\begin{recor}[Restatement of Corollary \ref{cor: two-period-two-arm}]
    In \emph{any} two‐period, two‐arm sequential experimental design, if \emph{sequential exchangeability}
  holds, then
  \[
    \tau_{t_2}^{\mathrm{seq}} \;=\; \tau.
  \]
\end{recor}

\begin{proof}
  Consider any two‐period, two‐arm design whose second‐period observed outcome admits the decomposition
  \[
    Y^{\mathrm{obs}}_{i,t_2}
    = \sum_{z_{t_1},z_{t_2}\in\{0,1\}}
      Y_{i,t_2}(z_{t_1},z_{t_2})
      \,1\{Z_{i,t_1}=z_{t_1}\}\,1\{Z_{i,t_2}=z_{t_2}\},
  \]
  as in Equation \ref{equation: t2_identity}. Under Assumption \ref{assumption: sequential_exchangeability} (sequential exchangeability), we have
  \[
    \mathbb{E}\bigl[Y^{\mathrm{obs}}_{i,t_2}\mid Z_{i,t_1}=z_{t_1}, Z_{i,t_2}=z_{t_2},\,X_{i,t_2}\bigr]
    = \mathbb{E}\bigl[Y_{i,t_2}(z)\mid X_{i,t_2}\bigr],
    \quad z_{t_1}, z_{t_2}\in\{0,1\},
  \]
  by (a) and (b) of Assumption \ref{assumption: sequential_exchangeability}, and hence
  \[
      \mathbb{E}\bigl[Y^{\mathrm{obs}}_{i,t_2}\mid Z_{i,t_2}=1,X_{i,t_2}\bigr]
      -\mathbb{E}\bigl[Y^{\mathrm{obs}}_{i,t_2}\mid Z_{i,t_2}=0,X_{i,t_2}\bigr] = \mathbb{E}\bigl[Y_{i,t_2}(1)-Y_{i,t_2}(0)\mid X_{i,t_2}\bigr]
  \]
  Finally, taking expectations over \(X_{i,t_2}\) and using the parallel‐trends assumption \(Y_{i,t_2}(z)=Y_{i,t_1}(z)+c\) from Assumption 2(c) yields
  \[
    \tau^{\mathrm{seq}}_{t_2}
    = E\bigl[Y_{i,t_2}(1)-Y_{i,t_2}(0)\bigr]
    = E\bigl[Y_{i,t_1}(1)-Y_{i,t_1}(0)\bigr]
    = \tau.
  \]
  This completes the proof.
\end{proof}

\subsection*{Control–Carryover Simulation}

We conducted a Monte Carlo study to assess how two distinct carryover structures bias the period-2 treatment effect when omitted from the model.  For each simulation:

\begin{itemize}
  \item \textbf{Sample size:} \(N = 100\) independent subjects.
  \item \textbf{Period 1:}
    \begin{align*}
      Z_1 &\sim \mathrm{Bernoulli}(0.5),\\
      X_1 &\sim \mathcal{N}(0,1),\\
      Y_1 &= \alpha_0 + \alpha_1 Z_1 + \alpha_2 X_1 + \varepsilon_1,
    \end{align*}
    with parameters
    \[
      \alpha_0 = 2,\quad \alpha_1 = 5,\quad \alpha_2 = 5,\quad \varepsilon_1 \sim \mathcal{N}(0,1).
    \]
  \item \textbf{Period 2 (two carryover structures):}
    \begin{enumerate}
      \item \emph{Additive interaction:}
        \[
          Y_{2}^{\rm interact} 
          = \beta_0 + \beta_1 Z_2 + \gamma\,(Z_1 \times Z_2) + \beta_2 X_2 + \varepsilon_2.
        \]
      \item \emph{Compounding effect:}
        \[
          Y_{2}^{\rm compound} 
          = \beta_0 + \beta_1^{\,1 + \gamma Z_1}\,Z_2 + \beta_2 X_2 + \varepsilon_2.
        \]
    \end{enumerate}
    Here,
    \[
      Z_2 \sim \mathrm{Bernoulli}(0.5),\quad
      X_2 = X_1 + \delta,\quad \delta\sim \mathcal{N}(0,1),
    \]
    and
    \[
      \beta_0 = 2,\quad \beta_1 = 5,\quad \beta_2 = 5,\quad \varepsilon_2 \sim \mathcal{N}(0,1).
    \]
  \item \textbf{Grid of carryover magnitudes:} \(\gamma\in\{-5, -4.9, \dots, +5\}\) (100 values), with 100 replications each.
  \item \textbf{Models fitted to \(Y_2\):}
    \begin{enumerate}
      \item \emph{Naïve:} \(\mathrm{lm}(Y_2 \sim Z_2 + X_2)\).
      \item \emph{Fixed-effect control:} \(\mathrm{lm}(Y_2 \sim Z_2 + Z_1 + X_2)\).
    \end{enumerate}
  \item \textbf{Summary:} For each model and carryover structure, we recorded the mean, minimum, and maximum of the estimated coefficient on \(Z_2\) across replications, and plotted these in Figure \ref{fig:simulation_control_carryover}.
\end{itemize}

\subsection*{CWSD Design–Comparison Simulation}

To compare bias under a fully counterbalanced within-subjects design versus sequential randomization, we simulated two-period data for \(N = 1{,}000\) subjects and evaluated five analytic strategies.

\begin{itemize}
  \item \textbf{True effect:} \(\tau = 1\).
  \item \textbf{Period 1:}
    \begin{align*}
      Z_{t1} &\sim \mathrm{Bernoulli}(0.5),\\
      X_{t1} &\sim \mathcal{N}(0,1),\\
      Y_{t1} &= \tau\,Z_{t1} + X_{t1} + \varepsilon_{t1},\quad \varepsilon_{t1}\sim \mathcal{N}(0,1).
    \end{align*}
  \item \textbf{Period 2:}
    \begin{itemize}
      \item \emph{Design A (Counterbalanced):} \(Z_{t2} = 1 - Z_{t1}\).
      \item \emph{Design B (Sequential Randomization):} \(Z_{t2}\sim \mathrm{Bernoulli}(0.5)\).
      \item Covariate drift:
        \[
          X_{t2} = X_{t1} + \gamma\,Z_{t1},\quad \gamma\in\{-10, -9.8, \dots, +10\}\ (\text{101 values}).
        \]
      \item Outcome:
        \[
          Y_{t2} = \tau\,Z_{t2} + X_{t2} + \varepsilon_{t2},\quad \varepsilon_{t2}\sim \mathcal{N}(0,1).
        \]
    \end{itemize}
  \item \textbf{Replications:} 100 replications per \(\gamma\), stacking \((Y_{t1},Y_{t2})\), \((Z_{t1},Z_{t2})\), and \((X_{t1},X_{t2})\).
  \item \textbf{Estimation strategies:}
    \begin{enumerate}
      \item \emph{No control:} \(\mathrm{lm}(Y \sim Z)\).
      \item \emph{Direct control:} \(\mathrm{lm}(Y \sim Z + X)\).
      \item \emph{Propensity-score adjustment:} fit \(\hat e(X) = \Pr(Z=1\mid X)\) via logistic regression, then \(\mathrm{lm}(Y \sim Z + \hat e)\).
      \item \emph{Fixed-effects:} include period indicator \(t\) in \(\mathrm{lm}(Y \sim Z + t)\).
      \item \emph{Misspecification:} control with random noise \(U \sim \mathcal{N}(0,1)\) via \(\mathrm{lm}(Y \sim Z + U)\).
    \end{enumerate}
  \item \textbf{Summary:} For each strategy and design, we computed the mean and range (min, max) of the estimated \(\tau\) across replications.  Results appear in Figure~\ref{fig:simulation_cwsd}.
\end{itemize}

\subsection*{Code Availability}
The code used for this project is publicly available at: \url{https://github.com/justinhjy1004/CWSD}

\end{document}